\documentclass{article}

\usepackage[margin=1in]{geometry}
\usepackage[onehalfspacing]{setspace}
\usepackage{amsmath,mathtools}
\usepackage{amssymb}
\usepackage{pgfplots}
\usepackage{subcaption}
\usepackage{booktabs}
\usepackage[braket, qm]{qcircuit}
\usepackage{graphicx}
\usepackage{authblk}
\usepackage{xcolor}
\usepackage{booktabs}
\usepackage{subcaption}
\usepackage{comment}
\usepackage[linesnumbered,ruled,vlined]{algorithm2e}
\usepackage{qcircuit}
\usepackage{pgffor}
\usepackage{ifthen}
\usepackage{tcolorbox}
\usepackage{longtable}
\usepackage{csquotes}

\SetCommentSty{mycommfont}

\SetKwInput{KwInput}{Input} %
\SetKwInput{KwOutput}{Output}  %
\SetKwProg{Fn}{Function}{}{}

\usepackage{siunitx} %

\usepackage{multicol}
\usepackage{centeredline}
\usepackage{nicematrix}
\usepackage{arydshln}

\usepackage{dsfont} 

\newcommand{\prob}[1]{P \left(#1\right)}
\newcommand{\indicator}[1]{\mathds{1}\left[#1\right]}

\usepackage[english]{babel}

\usepackage{pdflscape}
\usepackage{fancyvrb}
\usepackage{calc}
\usepackage{rotating}
\usepackage{pgf,tikz}
\usepackage{mathrsfs}
\usetikzlibrary{arrows}
\usepackage{mathtools}

\usepackage{listings}
\usepackage[]{qcircuit}
\usepackage{braket}
\usepackage{mathtools}
\usepackage{varwidth}
\usepackage{caption}
\usepackage{subcaption}
\usepackage{graphicx}
\usepackage[export]{adjustbox}

\usepackage{float}

\makeatletter
\def\paragraph{\@startsection{paragraph}{4}%
	\z@\z@{-\fontdimen2\font}%
	{\normalfont\bfseries}}
\makeatother

\usepackage{graphicx}
\usepackage{pgffor}
\usepackage{tikz,tikz-cd}
\usetikzlibrary{backgrounds,fit,decorations.pathreplacing}  %
\usepackage{booktabs}
\usepackage{enumerate}

\usepackage{amssymb, amsmath, amsthm}

\usepackage{xypic}%
\usepackage{scalerel}%
\usepackage{ulem}

\usepackage{graphicx}  %
\usepackage{amsmath}%
\usepackage{amsfonts}  %
\usepackage{amsthm} %
\usepackage{xkeyval}%
\usepackage{amssymb}
\usepackage{enumerate} %
\usepackage{color}
\usepackage{breqn}  %
\usepackage{bm}  %
\usepackage{cite}

\allowdisplaybreaks[1]
\usepackage{nicefrac}
\usepackage{booktabs}

\usepackage{comment}

\usepackage{kpfonts}
\usepackage{stackengine}
\usepackage{calc}
\newlength\shlength
\newcommand\xshlongvec[2][0]{\setlength\shlength{#1pt}%
	\stackengine{-5.6pt}{$#2$}{\smash{$\kern\shlength%
			\stackengine{7.55pt}{$\mathchar"017E$}%
			{\rule{\widthof{$#2$}}{.57pt}\kern.4pt}{O}{r}{F}{F}{L}\kern-\shlength$}}%
	{O}{c}{F}{T}{S}}

\usepackage{nicematrix}

\usepackage{enumitem}

\usepackage{hyperref}
\hypersetup{
	colorlinks= true, %
	citecolor = red%
}

\graphicspath{{./figures/}}

\newcommand{\RN}[1]{%
	\textup{\uppercase\expandafter{\romannumeral#1}}%
}

\allowdisplaybreaks

\newtheorem{thm}{Theorem}[section]

\newtheorem{lemma}[thm]{Lemma}

\newtheorem{remark}[thm]{Remark}

\def\<{\langle}
\def\>{\rangle}

\numberwithin{equation}{section}

\usetikzlibrary{matrix,fit}

\usepackage{pgfplots}
\pgfplotsset{compat=1.17}

\makeatletter
\def\smallunderbrace#1{\mathop{\vtop{\m@th\ialign{##\crcr
				$\hfil\displaystyle{#1}\hfil$\crcr
				\noalign{\kern3\p@\nointerlineskip}%
				\tiny\upbracefill\crcr\noalign{\kern3\p@}}}}\limits}
\makeatother

\newcommand{\AWQPE}{Adaptive Windowed Quantum Phase Estimation}
\newcommand{\AWQPEs}{AWQPE}
\newcommand{\mList}{[m_1, m_2, \dots, m_B]} 
\newcommand{\nTotal}{n} 
\newcommand{\currentM}{m_i} 
\newcommand{\phiEst}{\phi_{\text{est}}}
\newcommand{\phiRaw}{\phi_{\text{raw}}}

\newcommand{\AmbFlags}{\mathcal{A}} 
\newcommand{\LastIdx}{\mathcal{S}_{\text{idx}}} 
\newcommand{\NumShots}{N_{\text{shots}}}
\newcommand{\NumTarg}{n_T}
\newcommand{\UMat}{U}
\newcommand{\Eigenstate}{\ket{u}}

\newcommand{\Threshold}{\epsilon}
\newcommand{\CurrentIter}{i}
\newcommand{\TotalBitsEstimated}{k}
\newcommand{\MostLikelyOutcome}{\mathbf{b}_{\text{ml}}}

\newcommand{\FlagAmb}{\text{flag}_{\text{amb}}}

\newcommand{\Frac}[1]{\left\{ #1 \right\}}

\newcommand{\bestApprox}[2]{\widehat{#1}_{#2}}

\makeatletter
\def\algocf@startfloat{
 \let\orig@float\relax
 \ifx\orig@float\undefined
  \let\orig@float\float
 \fi
 \orig@float
}
\makeatother

\SetKwInOut{Input}{Input}
\SetKwInOut{Output}{Output}

\begin{document}
 \title{Towards Practical Quantum Phase Estimation:\\ A Modular, Scalable, and Adaptive Approach}
	\author[1]{Alok Shukla \thanks{Corresponding author.}}
	\author[2]{Prakash Vedula}
	\affil[1]{School of Arts and Sciences, Ahmedabad University, India}
	\affil[1]{alok.shukla@ahduni.edu.in}
	\affil[2]{School of Aerospace and Mechanical Engineering, University of Oklahoma, USA}
	\affil[2]{pvedula@ou.edu}
	
\date{}

	\maketitle

\begin{abstract}
Quantum Phase Estimation (QPE) is a cornerstone algorithm in quantum computing, with applications ranging from integer factorization to quantum chemistry simulations. However, the resource demands of standard QPE, which require a large number of coherent qubits and deep circuits, pose significant challenges for current Noisy Intermediate Scale Quantum (NISQ) devices.
In this work, we introduce the Adaptive Windowed Quantum Phase Estimation (AWQPE) algorithm, a novel method designed to address the limitations of standard QPE. AWQPE utilizes small, independent blocks of $m > 1$ control qubits to estimate multiple phase bits simultaneously within a ``window,'' thereby significantly reducing the number of iterations required to achieve a desired precision. These independent blocks are amenable to parallelization and, when combined with a robust least-significant-bit (LSB) to most-significant-bit (MSB) ambiguity resolution mechanism, enhance the algorithm's accuracy while mitigating the risk of error propagation. 
Our numerical simulations demonstrate AWQPE’s accuracy and robustness, showcasing a distinct balance between resource efficiency and computational speed. This makes AWQPE particularly well-suited for near-term quantum platforms. 
\end{abstract}

\maketitle

\section{Introduction}
\label{sec:introduction}

Quantum Phase Estimation (QPE) is a fundamental algorithm in quantum computation with broad applications, including Shor's algorithm for integer factorization and solving linear systems of equations \cite{nielsen2002quantum, shor1994algorithms, harrow2009quantum}. Given a unitary operator $U$ and an initial quantum eigenstate $\ket{u}$, QPE aims to estimate the corresponding eigenphase $\phi$ where $U\ket{u} = e^{2\pi i \phi}\ket{u}$.

The standard QPE algorithm, while theoretically powerful, is characterized by its demanding resource requirements, making it largely infeasible for current Noisy Intermediate-Scale Quantum (NISQ) devices. This approach requires a large register of $n$ coherent ancilla qubits to achieve a high $n$-bit precision, and its circuit depth scales poorly due to the high powers of $U$ and the Inverse Quantum Fourier Transform (IQFT). 

To address these limitations, many alternative QPE methods have been explored in the literature, with a focus on improving efficiency and robustness against noise \cite{knill2007optimal, berry2015simulating, poulin2009sampling, wang2022quantum, dong2022ground, wan2022randomized, zhang2022computing}. Iterative phase estimation (IPE) methods, such as Kitaev's algorithm \cite{kitaev1995quantum}, were among the first to offer a more hardware-friendly alternative. These algorithms estimate the phase bit-by-bit using a simple circuit based on the Hadamard test with dyadic powers $U^{2^j}$. The sequential nature of IPE significantly reduces the circuit depth and qubit count, but it is susceptible to error accumulation, where inaccuracies in early bit estimations can propagate and corrupt subsequent results. Furthermore, these methods often require a large number of measurement shots per iteration to achieve statistical certainty, leading to a considerable total runtime.

More recent iterative methods, such as Robust Phase Estimation (RPE) or its variants \cite{kimmel2015robust,belliardo2020achieving,russo2021evaluating,ni2023low}, build upon this framework by using a single ancilla qubit and statistical techniques to mitigate noise. However, their bit-by-bit estimation still limits their efficiency for high-precision tasks.

In this paper, we introduce the Adaptive Windowed Quantum Phase Estimation (AWQPE) algorithm, a novel method that extends the practical benefits of iterative techniques while overcoming their key limitations. Our approach utilizes small, independent fixed blocks, each with $m_j>1$ control qubits, to estimate multiple phase bits simultaneously in a ``window." These independent blocks also allow for parallelized estimation. The robust classical post-processing logic (for ambiguity resolution) of AWQPE, enables it to achieve a high degree of precision with significantly lower number of qubits,  circuit depth and gate complexity per block. The flexibility to choose a small window size  maintains a low qubit count, while the reduced number of iterations mitigates the risks of error accumulation and minimizes total runtime. This strategy provides a distinct balance between resource efficiency and computational speed, making AWQPE particularly well-suited for NISQ platforms.

The remainder of this paper is organized as follows: Section \ref{sec:algo} details the proposed AWQPE algorithm. A proof of correctness of the AWQPE algorithm is shown in Section \ref{sec:proof}.
A discussion on computational complexity of the AWQPE algorithm is presented in Section \ref{sec:complexity}.
Section \ref{sec:advantages} explicitly discusses the advantage of our proposed approach. Section \ref{sec:results} presents numerical simulation results. Finally, Section \ref{sec:conclusion} concludes with a summary of the presented work.

\section{\AWQPE \, (\AWQPEs) }
\label{sec:algo}

The Adaptive Windowed Quantum Phase Estimation (AWQPE) algorithm is designed for  estimating (with a high degree of precision) the phase $\phi$ of an eigenstate $\Eigenstate$ of a unitary operator $U$, which acts on $\NumTarg$ target qubits. This method employs a windowed approach, characterized by a bit allocation strategy denoted as $ \mList$, where the total number of bits to be estimated is $\nTotal = \sum_{i=1}^{B} m_i$. Here the $j$-th block is allocated $m_j$ qubits and each block is independently processed by the AWQPE algorithm. An important component of AWQPE is its classical post-processing step, specifically engineered for robust ambiguity resolution.

For clarity in handling modular arithmetic and comparisons, particularly in the context of phase estimation where $0$ and $n-1$ (for $n$ discrete values) are adjacent, we define a modified minimum function. For distinct $a, b \in \{0, 1, 2, \ldots, n-1\}$, the minimum operation modulo $n$ is defined as:
\begin{align} \label{eq:minabmod}
\min(a, b) \bmod n =
\begin{cases}
\min(a, b) & \text{if } a < b,\ a \notin \{0,n-1\},\ b \notin \{0,n-1\}, \\
n - 1 & \text{if } (a = 0 \text{ and } b = n - 1) \text{ or } (a = n - 1 \text{ and } b = 0).
\end{cases}
\end{align}
This definition ensures that when 0 and $n-1$ are considered, the result appropriately reflects their proximity in $\mathbb{Z}/n \mathbb{Z}$.

The structure of the AWQPE approach is presented in Algorithm \ref{alg:wiqpe_multi_qubit}. The process begins with an initialization phase where variables are set: $\phiRaw$ is initialized as an empty string to accumulate raw binary estimates, $\nTotal$ is calculated as the sum of all chunk sizes, $\TotalBitsEstimated$ tracks the total number of bits estimated so far, $\AmbFlags$ is an empty list to record ambiguity statuses for each chunk, and $\CurrentIter$ is set to zero. An ambiguity threshold $\Threshold$, typically around $0.9$, is established to identify potentially ambiguous measurement outcomes.

The core of Algorithm \ref{alg:wiqpe_multi_qubit} is an iterative loop that continues until all $\nTotal$ bits of the phase have been estimated. In each iteration, a quantum circuit is prepared to estimate a specific segment of the phase, corresponding to the current window size, $\currentM$, from the list $\mList$. This circuit involves allocating $\currentM$ control qubits and $\NumTarg$ target qubits, initializing the target qubits in the eigenstate $\Eigenstate$, and applying Hadamard gates to all control qubits. Subsequently, a sequence of controlled unitary operations $U^{2^{\TotalBitsEstimated + p}}$ are applied, where $p$ ranges from $0$ to $\currentM-1$, to the target qubits, controlled by the respective control qubits. Following these operations, an inverse Quantum Fourier Transform (IQFT) is applied to the control qubits, which are then measured in the computational basis.

The circuit is executed $\NumShots$ times, and the frequencies of all possible outcome strings of length $\currentM$ are recorded. The algorithm then identifies the two most likely outcomes, $t_1^*$ and $t_2^*$, based on their respective counts, $C(t_1^*)$ and $C(t_2^*)$. Initially, $t_1^*$ is selected as the chunk estimate (i.e., the most likely binary string estimate for the $i$-th chunk), $\MostLikelyOutcome$. An ambiguity check is then performed: if the ratio of the count of the second most likely outcome to the most likely outcome ($C(t_2^*)/C(t_1^*)$) exceeds the predefined threshold $\Threshold$, an ambiguity flag, $\FlagAmb$, is set to True. If ambiguity is detected and the current block is not the final block (i.e., $\TotalBitsEstimated + \currentM < \nTotal$), the chunk estimate $\MostLikelyOutcome$ is updated to $\min(t_1^*, t_2^*) \bmod 2^{\currentM}$, adhering to the modular minimum definition in Eq. \ref{eq:minabmod}. The $\FlagAmb$ is appended to $\AmbFlags$ (which is a list that stores the ambiguity flag for each chunk), $\MostLikelyOutcome$ is concatenated to $\phiRaw$, and $\TotalBitsEstimated$ and $\CurrentIter$ are updated for the next iteration.

Upon completion of all iterations, a classical post-processing routine, \texttt{AWQPEAmbiguityResolution}, is invoked. This routine, detailed in Algorithm \ref{alg:wiqpe_ambiguity_resolution}, takes raw estimates $\phiRaw$, $\mList$, and $\AmbFlags$ as input to produce the final corrected binary phase estimate, $\phiEst$, and an index $\LastIdx$, indicating a special chunk if present.

Algorithm \ref{alg:wiqpe_ambiguity_resolution} functions by first partitioning the concatenated raw binary string $\phiRaw$ into its original $m_j$-bit chunks, denoted as $[\phi^{(1)}, \phi^{(2)}, \dots, \phi^{(B)}]$. A preliminary step identifies the rightmost non-zero chunk that corresponds to the integer value $2^{\currentM - 1}$ (binary `10...0'). If such a chunk is found, its index is stored in $\LastIdx$. This identification is critical for specific correction rules.

The core of the ambiguity resolution (handled via Algorithm \ref{alg:wiqpe_ambiguity_resolution}) involves applying Least Significant Bit (LSB)-to-Most Significant Bit (MSB) corrections. This process iterates backward from the second-to-last chunk ($j = B-1$) to the first chunk ($j = 1$). In each step, a correction bit,  $b_{\mathrm{corr}}$, is initially determined by the MSB of the next less significant chunk, $\phi^{(j+1)}$. However, a key aspect of the algorithm's robustness is its conditional correction: if the current chunk $\phi^{(j)}$ was marked as ambiguous in $\AmbFlags[j-1]$ (adjusting for 0-indexing of $\AmbFlags$ versus 1-indexing of chunks), or if the immediately less significant chunk $\phi^{(j+1)}$ was identified as the special chunk (i.e., $\LastIdx = j+1$), then $b_{\mathrm{corr}}$ is set to 0, effectively bypassing the correction for that specific chunk. Otherwise, the integer value $x$ of the current chunk $\phi^{(j)}$ is updated to $x_{\mathrm{new}} \gets (x - b_{\mathrm{corr}}) \bmod 2^{m_j}$, applying the correction modulo $2^{m_j}$. The corrected integer value $x_{\mathrm{new}}$ is then converted back into its $m_j$-bit binary string representation and replaces the original $\phi^{(j)}$. Finally, after all corrections have been applied, the individual corrected chunks $\phi^{(1)}, \phi^{(2)}, \dots, \phi^{(B)}$ are concatenated to form the final corrected binary phase estimate, $\phiEst$.

\begin{algorithm}[H]
\small
\SetAlgoLined
\SetKwInOut{Input}{Input}
\SetKwInOut{Output}{Output}

\Input{
Bit allocation per block $[m_1, m_2, \dots, m_B]$ with $m_i > 1$;\\
$\UMat$: unitary operator acting on $2^{\NumTarg}$-dimensional space;\\
$\Eigenstate$: eigenstate of $\UMat$;\\
$\NumTarg$: number of target qubits;\\
$\NumShots$: number of measurement repetitions per iteration.\\
}

\Output{
$\phiEst$: final binary estimate of the phase;\\
$\phiRaw$: concatenated raw estimate from each window;\\
$\AmbFlags$: list of ambiguity flags for each chunk;\\
$\LastIdx$: index of the special chunk (if any), otherwise \texttt{`None'}.
}

\textbf{Initialization:} \\
$\phiRaw \gets \texttt{""}$, $\nTotal \gets \sum_{j=1}^{B} m_j$, 
$\TotalBitsEstimated \gets 0$, 
$\AmbFlags \gets [\ ]$, 
$\CurrentIter \gets 0$,
Set an ambiguity threshold $\Threshold$ (e.g., $0.9$); \\

\While{$\TotalBitsEstimated < \nTotal$}{
 $\currentM \leftarrow $ the current window size for block $i$; \tcp{Current window size from $\mList$.}
 
 Allocate $\currentM$ control qubits and $\NumTarg$ target qubits;\\
 Initialize target qubits in eigenstate $\Eigenstate$;\\
 Apply Hadamard gates to all control qubits;\\
 \For{$p \gets 0$ \KwTo $\currentM - 1$}{
  Apply controlled unitary $U^{2^{\TotalBitsEstimated + p}}$ to the target qubits controlled on $p$-th control qubit;
 }
 Apply inverse Quantum Fourier Transform (IQFT) to control qubits;\\
 Measure the control qubits in the computational basis;

 Execute the circuit $\NumShots$ times and record the frequency of each outcome string of length $\currentM$
 \tcc*{Determine most likely outcomes. In case of tie, any choice is picked randomly.} 
  $C(t_1^{*}) \leftarrow $ the count of the most likely outcome, i.e., $C(t_1^*) \geq C(t)$ for $t \in \{0,\dots, 2^{m_i} -1\}$, where $C(t)$ is the count of observing the state $t$ on measurement;  
  \\
 $C(t_2^{*}) \leftarrow $ the count of the most likely outcome, i.e., $C(t_2^*) \geq C(t)$ for $t \in \{0,\dots, 2^{\currentM} -1\} - \{t_1^*\}$; \\
 $\MostLikelyOutcome \leftarrow t_1^*$; \tcp{Select $t_1^*$ as the chunk estimate.}
  \tcc{Check for ambiguity.}
  $\FlagAmb \gets \text{False}$;\\
 \If{({${C(t_{2}^*)}/ {C(t_1^*)} > \Threshold$}) } 
 {
  $\FlagAmb \gets \text{True}$;\\
  \If{($\TotalBitsEstimated + \currentM < \nTotal$)}
  {
  $\MostLikelyOutcome \leftarrow  \min(t_1^{*}, t_2^{*}) \mod 2^{\currentM}$; \tcp{Select $\min(t_1^*, t_2^*) \bmod 2^{\currentM}$ as the chunk estimate. Refer Eq.~\ref{eq:minabmod}.}
  }
 }
 Append $\FlagAmb$ to $\AmbFlags$;\\
 $\phiRaw \leftarrow \text{ concatenate } \phiRaw \circ \MostLikelyOutcome$; \tcp{Here '$\circ$' operation represents concatenation of strings.}
 $\TotalBitsEstimated \gets \TotalBitsEstimated + \currentM$, 
 $\CurrentIter \gets \CurrentIter + 1$;
}

Apply a post-processing routine \texttt{AWQPEAmbiguityResolution}$(\phiRaw, \mList, \AmbFlags)$ to obtain 
$\phiEst$ and $\LastIdx$ \tcp*{Refer to Algorithm~\ref{alg:wiqpe_ambiguity_resolution}.}

\KwResult{$\phiEst$, $\phiRaw$, $\AmbFlags$, $\LastIdx$.}
\caption{\AWQPE \, (\AWQPEs) with Multi-Qubit Target}
\label{alg:wiqpe_multi_qubit}
\end{algorithm}

\begin{algorithm}[H]
\small
\SetAlgoLined
\SetKwInOut{Input}{Input}
\SetKwInOut{Output}{Output}

\Input{
$\phiRaw$: Concatenated raw estimate from Algorithm~\ref{alg:wiqpe_multi_qubit}; \\
$ [m_1, m_2, \dots, m_B]$: Bit allocation per block;\\
$\AmbFlags$: List of ambiguity flags (MSB to LSB order).
}

\Output{
$\phiEst$: Final corrected binary phase estimate;\\
$\LastIdx$: Index of special chunk (if any), otherwise \texttt{`None'}.
}

\BlankLine
\tcp{Split concatenated string into $m_j$-bit chunks}
Partition $\phiRaw$ into substrings $[\phi^{(1)}, \phi^{(2)}, \dots, \phi^{(B)}]$ such that the size of the $j$-th block is $m_j$, i.e., $|\phi^{(j)}| = m_j$;\\
$\LastIdx \gets \texttt{`None'}$; 

\BlankLine
\tcp{Step 1: Identify rightmost non-zero chunk equal to $2^{m_j - 1}$ (i.e., '10...0')}
\For{$j \leftarrow B$ \KwTo $1$}{
 Let $x \gets$ integer value of $\phi^{(j)}$;\\
 \If{$x = 0$}{Continue;}
 \Else{
  \If{$x = 2^{m_j - 1}$}{
$\LastIdx \leftarrow j$; \\
}
 \textbf{break}
  }
 }

\BlankLine
\tcp{Step 2: Apply LSB-to-MSB corrections}
\For{$j \gets B - 1$ \KwTo $1$}{
 $b_{\mathrm{corr}} \gets$ most significant bit of $\phi^{(j+1)}$; \tcp{Default correction from MSB of next chunk}

 \If{($\AmbFlags[j] = \text{True}$) or ($\LastIdx = j+1$)}{
  $b_{\mathrm{corr}} \gets 0 $; \tcp{For ambiguous chunk, or if next chunk is special, no correction needed.}
 }

 Let $x \gets$ integer value of $\phi^{(j)}$;\\
 $x_{\mathrm{new}} \gets (x - b_{\mathrm{corr}}) \mod 2^{m_j}$; \tcp{Apply correction modulo $2^{m_j}$}
 $\phi^{(j)} \gets$ binary string of $x_{\mathrm{new}}$ with $m_j$ bits;
}

\BlankLine
\tcp{Step 3: Reconstruct corrected binary phase}
$\phiEst \gets$ concatenate $\phi^{(1)} \circ \phi^{(2)} \circ \cdots \circ \phi^{(B)}$;

\KwResult{$\phiEst$, $\LastIdx$}
\caption{\texttt{AWQPEAmbiguityResolution} - \AWQPEs \, Ambiguity Resolution (LSB-to-MSB Post-Processing)}
\label{alg:wiqpe_ambiguity_resolution}
\end{algorithm} 
\vspace{0.5cm}

\begin{remark}[Special Chunk]
    The `special chunk' ($\LastIdx$) is designed to address the exceptionally rare cases where the true fractional part of the phase on the right of a chunk boundary is exactly $0.5$. In such a scenario, the standard LSB-to-MSB correction, which relies on the fractional part being strictly greater than or less than $0.5$, becomes ambiguous. While in standard QPE this typically introduces only a small rounding ambiguity, in a windowed iterative scheme like AWQPE, an unhandled $0.5$ ambiguity could, in principle, propagate errors. However, this edge case is statistically highly improbable. Furthermore, its potential impact can be mitigated by slight adjustments to the bit allocation $m_i$ or by introducing a small perturbation to the phase (see Remark \ref{remark:perturb}), ensuring the fractional part is not precisely $0.5$. Even in these specific `special chunk' instances where the ambiguity resolution mechanism bypasses the correction due to this exact $0.5$ condition, the quantum measurement itself will yield one of two possible results for the problematic bit with approximately 50\% probability. This implies that there is still a 50\% chance that the uncorrected bit (due to this specific bypass) is, in fact, the correct one.
\end{remark}

\begin{remark} [Enhanced Confidence and Ambiguity Handling]\label{remark:perturb}
    While AWQPE is inherently robust, an additional procedure can be employed to achieve absolute confidence in the estimated phase, particularly to explicitly handle the rare `special chunk' cases. This involves performing the AWQPE procedure twice. The first run is for the original unitary operator $U$ to estimate $\phi$. If this run identifies a `special chunk' (indicated by $\LastIdx$), the nature of this special chunk provides precise information about the problematic $0.5$ fractional part. A perturbation $\Delta\phi$ can then be strategically chosen for the second run, such that the problematic fractional part is explicitly shifted away from the $0.5$ ambiguity point. The second run is then performed for a perturbed unitary operator $U' = e^{i2\pi\Delta\phi} U$ to estimate $\phi + \Delta\phi$, where $\Delta\phi$ could be implemented via appropriate additional phase gates. Let the two estimated phases be $\phi_{\text{est}}$ and $\phi'_{\text{est}}$. Absolute confidence can then be established by verifying that $\phi'_{\text{est}} - \phi_{\text{est}} \approx \Delta\phi \pmod{2\pi}$. A significant deviation from the expected $\Delta\phi$ would immediately flag a potential issue. This method provides a powerful cross-validation with a targeted approach to ambiguities.
\end{remark}

\begin{remark}[Alternative Ambiguity Resolution Criteria]
While our ambiguity resolution procedure relies on the condition ${C(t_{2}^*)}/ {C(t_1^*)} > \Threshold$,  other information theoretic approaches may also be employed. For example, Shannon entropy can serve as a global measure of uncertainty, flagging ambiguous distributions when entropy exceeds a given threshold. Similarly, Kullback–Leibler (KL) divergence between the observed distribution and an ideal (for example, peaked or uniform) reference distribution can provide a quantitative diagnostic for ambiguity. 
\end{remark}

\begin{remark}[Relevance to Shor's Algorithm]
The core framework of AWQPE algorithm can be effectively applied to integer factorization by substituting the QPE block in Shor’s algorithm. We have validated its feasibility through preliminary experiments, with a detailed account to appear in a forthcoming manuscript.

\end{remark}

Having detailed the procedural steps of the AWQPE algorithm and its classical ambiguity resolution routine (Algorithms \ref{alg:wiqpe_multi_qubit} and \ref{alg:wiqpe_ambiguity_resolution}), we now proceed to rigorously establish its correctness. The efficacy of the LSB-to-MSB correction mechanism, in particular, relies on fundamental properties of binary approximations and modular arithmetic. To this end, we first introduce a set of essential lemmas that underpin the algorithm's mathematical foundation.

\section{Proof of Correctness of \AWQPEs \, Algorithm}
\label{sec:proof}

\subsection{Useful Lemmas for Reconstruction of Estimated Phase}

To establish the theoretical guarantees of the \AWQPEs \, algorithm, particularly regarding the reconstruction of the estimated phase from its binary chunks, we first define the following preliminaries and notation.

\subsection*{Preliminaries and Notation}
\begin{itemize}[noitemsep]
 \item Let \( x \in [0, 1) \) be a real number and let \( m, k \in \mathbb{Z}^+ \) denote positive integers.
 \item An \( n \)-bit binary fraction is a number of the form \( I/2^n \), where \( I \in \mathbb{Z} \).
  \item The ``best $n$-bit approximation'' of a real number $y \in [0, 1)$, denoted \( \bestApprox{y}{n} \),  is defined as the $n$-bit binary fraction $I/2^n$ that minimizes $|y - I/2^n|$. Formally, $\bestApprox{y}{n}  = I = \lfloor y \cdot 2^n + 0.5 \rfloor$.
 \item Define \( \bestApprox{x}{m} = b_m / 2^m \), with numerator \( b_m = \lfloor x \cdot 2^m + 0.5 \rfloor \).
 \item Let \( x' = 2^m x \in [0, 2^m) \) and define \( \delta = \Frac{x'} = x' - \lfloor x' \rfloor \), the fractional part of \( x' \).
 \item Let \( b_k = \lfloor \delta \cdot 2^k + 0.5 \rfloor \), and define the best \( k \)-bit approximation to \( \delta \) as \( \bestApprox{\delta}{k} = b_k / 2^k \).
 \item We assume that \( \bestApprox{\delta}{k} \ne 0.5 \), i.e., \( b_k \ne 2^{k-1} \), to avoid ambiguity in rounding.
\end{itemize}

\begin{lemma}
\label{lem:delta_gt_half}
Let \( \delta \in [0,1) \), and define \( \bestApprox{\delta}{k} = \frac{\lfloor \delta \cdot 2^k + 0.5 \rfloor}{2^k} \). If \( \bestApprox{\delta}{k} > 0.5 \), then \( \delta > 0.5 \).
\end{lemma}

\begin{proof}
Assume \( \bestApprox{\delta}{k} > 0.5 \). Then:
\[
\frac{\lfloor \delta \cdot 2^k + 0.5 \rfloor}{2^k} > 0.5.
\]
Multiplying both sides by \( 2^k \), we obtain:
\[
\lfloor \delta \cdot 2^k + 0.5 \rfloor > 2^{k-1}.
\]
Since the left-hand side is an integer, it follows that
\[
\lfloor \delta \cdot 2^k + 0.5 \rfloor \ge 2^{k-1} + 1.
\]
By the definition of the floor function,
\[
\delta \cdot 2^k + 0.5 \ge 2^{k-1} + 1 \quad \Rightarrow \quad \delta \cdot 2^k \ge 2^{k-1} + 0.5.
\]
Dividing both sides by \( 2^k \) yields:
\[
\delta \ge \frac{2^{k-1} + 0.5}{2^k} = \frac{1}{2} + \frac{1}{2^{k+1}} > \frac{1}{2}.
\]
Thus, \( \delta > 0.5 \), as claimed.
\end{proof}

\begin{lemma}
\label{lem:combined_approx}
Let \( x \in [0, 1) \) be a real number. Further, let \( \bestApprox{x}{m} = b_m / 2^m \), with \( b_m = \lfloor x \cdot 2^m + 0.5 \rfloor \), and \( \delta = \Frac{2^m x} \). Let \( b_k = \lfloor \delta \cdot 2^k + 0.5 \rfloor \), and define \( \bestApprox{\delta}{k} = b_k / 2^k \), assuming \( \bestApprox{\delta}{k} \ne 0.5 \). Then, the best \((m+k)\)-bit approximation to \( x \) is given by:
\[
\hat{x}_{m+k} = \frac{B}{2^{m+k}}, \quad \text{where } B = p \cdot 2^k + b_k,
\]
with \( p \in \mathbb{Z} \) defined as:
\begin{enumerate}[label=(\alph*), itemsep=0pt]
 \item If \( \bestApprox{\delta}{k} > 0.5 \), then \( p = b_m - 1 \).
 \item If \( \bestApprox{\delta}{k} < 0.5 \), then \( p = b_m \).
\end{enumerate}
\end{lemma}

\begin{proof}
By definition, \( \bestApprox{x}{m} = b_m / 2^m \), with \( b_m = \lfloor x \cdot 2^m + 0.5 \rfloor \), which implies:
\[
b_m - 0.5 \le x \cdot 2^m < b_m + 0.5.
\]
Let \( x' = 2^m x \). We can express \( x' \) as \( x' = b_m + \eta \), where \( \eta = x' - b_m \in [-0.5, 0.5) \).

The fractional part of \( x' \) is \( \delta = \Frac{x'} = x' - \lfloor x' \rfloor \). The relationship between \( \delta \) and \( \eta \) depends on the sign of \( \eta \):
\begin{itemize}[noitemsep]
 \item If \( \eta \ge 0 \) (i.e., \( x' \in [b_m, b_m + 0.5) \)), then \( \lfloor x' \rfloor = b_m \), and \( \delta = x' - b_m = \eta \).
 \item If \( \eta < 0 \) (i.e., \( x' \in [b_m - 0.5, b_m) \)), then \( \lfloor x' \rfloor = b_m - 1 \), and \( \delta = x' - (b_m - 1) = \eta + 1 \). In this case, \( \delta \in [0.5, 1) \).
\end{itemize}
Next, we have
\[
x \cdot 2^{m+k} = x' \cdot 2^k = (\lfloor x' \rfloor + \delta) \cdot 2^k.
\]
Hence, the exact numerator of \( \bestApprox{x}{m+k} \) is
\[
B_{\text{true}} = \lfloor x \cdot 2^{m+k} + 0.5 \rfloor = \lfloor \lfloor x' \rfloor \cdot 2^k + \delta \cdot 2^k + 0.5 \rfloor.
\]
Using the definition of \( b_k = \lfloor \delta \cdot 2^k + 0.5 \rfloor \), we obtain
\[
B_{\text{true}} = \lfloor x' \rfloor \cdot 2^k + b_k.
\]
We now determine \( \lfloor x' \rfloor \) based on the value of \( \delta \).

\noindent \textbf{Case (a):} \( \bestApprox{\delta}{k} > 0.5 \). Then by Lemma~\ref{lem:delta_gt_half}, \( \delta > 0.5 \). This occurs when \( \eta < 0 \), so \( x' \in [b_m - 0.5, b_m) \). Thus,
\[
\lfloor x' \rfloor = b_m - 1 \quad \Rightarrow \quad B_{\text{true}} = (b_m - 1) \cdot 2^k + b_k.
\]
Therefore, \( p = b_m - 1 \).

\noindent \textbf{Case (b):} \( \bestApprox{\delta}{k} < 0.5 \). Then \( \delta < 0.5 \) (which can be proved using an argument similar to that of Lemma~\ref{lem:delta_gt_half} ), which corresponds to \( \eta \ge 0 \), and hence \( x' \in [b_m, b_m + 0.5) \). So,
\[
\lfloor x' \rfloor = b_m \quad \Rightarrow \quad B_{\text{true}} = b_m \cdot 2^k + b_k.
\]
Therefore, \( p = b_m \).

\textbf{Excluded Case:} If \( \bestApprox{\delta}{k} = 0.5 \), we would have \( \delta = 0.5 \), which may lead to rounding ambiguity. However, this case is explicitly excluded by assumption.

In both admissible cases, the constructed numerator \( B = p \cdot 2^k + b_k \) coincides with \( B_{\text{true}} \), completing the proof.
\end{proof}

\subsection{Probabilistic Guarantees for Most Likely Outcome in \AWQPEs}

Let $\delta \in [0,1)$ be the true fractional phase for a window of $m$ bits. Let $\NumShots$ be the number of measurement repetitions for this window. The probability of measuring an outcome $j \in \{0, 1, \dots, 2^m-1\}$ is given by $P(j|\delta)$, which for the ideal standard QPE output distribution is the squared Dirichlet kernel,
$$P(j|\delta) =
\frac{1}{2^{2m}}\cdot\frac{\sin^2(2^m \pi \theta )}{\sin^2(\pi \theta)},$$
where $\theta = \delta-\frac{j}{2^m}$.
The true most probable outcome, $T_1$, is the integer $j \in \{0, \dots, 2^m-1\}$ that minimizes $|2^m\delta - j|$. This is unique unless $2^m\delta$ is exactly a half-integer. For a non-ambiguous phase, $T_1 = \lfloor 2^m \delta + 0.5 \rfloor \pmod{2^m}$. The true second most probable outcome, $T_2$, is the integer adjacent to $T_1$ with the second highest probability. Two integers $j_1, j_2 \in \{0, \dots, 2^m-1\}$ are defined as adjacent if $j_1-j_2  \equiv \pm 1 \pmod{2^m} $. Let $N_j$ be the number of times outcome $j$ is measured in $\NumShots$ trials.

\begin{lemma}[Probabilistic Outcome for a \AWQPEs \, Window]
\label{lem:probabilistic_outcome_quantitative}
For any iteration of the \AWQPE \, (\AWQPEs) algorithm, which aims to estimate $m$ bits of a fractional phase $\delta$, the following probabilistic guarantees hold.
\end{lemma}

\begin{enumerate} [label=(\alph*), itemsep=0pt]
 \item \textbf{High Probability of Correctly Identifying the Most Probable Outcome}

 For any desired error probability $\epsilon_1 > 0$, the observed most frequent outcome $t_1^*$ is equal to the true most probable outcome $T_1$ with a probability of at least $1-\epsilon_1$. This holds provided that the number of measurements $\NumShots$ satisfies:
 $$ \NumShots \ge \frac{2}{\Delta P_{\min}^2} \ln\left(\frac{2^m- 2}{\epsilon_1}\right),$$
 where $\Delta P_{\min} = P(T_1|\delta) - \max_{j \notin \{T_1, T_2\}} P(j|\delta)$ is the minimum probability gap between the most probable outcome and any other non-adjacent outcome. Further, $\Delta P_{\min}$ is bounded below by a positive number.

 \begin{proof}
 The proof proceeds in two steps. First, we bound the probability of a single incorrect outcome having a higher count than $T_1$. Second, we use a union bound to generalize this to all possible incorrect outcomes.

 \textbf{Step 1: Bounding a single-pair error probability.}
 The observed most frequent outcome $t_1^*$ is incorrect if the count $N_{T_1}$ is not the highest count. This happens if there exists at least one other outcome $j \ne T_1$ for which $N_j \ge N_{T_1}$.
 
 Consider a specific incorrect outcome $j \ne T_1$. We want to bound the probability $\prob{N_j \ge N_{T_1}}$. Let $p_j = P(j|\delta)$ be the true probability of outcome $j$. We know $p_{T_1} > p_j$ for all $j \ne T_1$ in the non-ambiguous case.
 
 For each measurement $k \in \{1, \dots, \NumShots\}$, we define an independent and identically distributed (i.i.d.) random variable $Y_k$ using indicator functions:
 $$ Y_k = \indicator{\text{outcome } k \text{ is } T_1} - \indicator{\text{outcome } k \text{ is } j} .$$
 The random variable of interest is the difference in counts, $Z = N_{T_1} - N_j = \sum_{k=1}^{\NumShots} Y_k$. Each $Y_k$ takes values in $\{-1, 0, 1\}$ and has an expected value $E[Y_k] = p_{T_1} - p_j = \Delta_{T_1,j}$. The range of $Y_k$ is $1 - (-1) = 2$.
 
 We are interested in the probability of a large deviation, specifically $\prob{N_j \ge N_{T_1}} = \prob{Z \le 0}$. The expected value of $Z$ is $E[Z] = \NumShots E[Y_k] = \NumShots \Delta_{T_1,j}$. Applying Hoeffding’s inequality for the sum of bounded i.i.d. variables:
 \begin{align*}
 \prob{Z \le 0} &= \prob{Z - E[Z] \le -E[Z]} \\
 &\le \exp\left(-\frac{2(E[Z])^2}{\NumShots \cdot (1 - (-1))^2}\right) \\
 &= \exp\left(-\frac{2(\NumShots \Delta_{T_1,j})^2}{4\NumShots}\right) \\
 &= \exp\left(-\frac{\NumShots \Delta_{T_1,j}^2}{2}\right).
 \end{align*}
 Thus, for any specific outcome $j \ne T_1$, the probability of it being observed more frequently than $T_1$ is bounded by $\prob{N_j \ge N_{T_1}} \le \exp(-\NumShots \Delta_{T_1,j}^2/2)$.

 \textbf{Step 2: Applying a union bound.}
 The total error probability is the probability that at least one incorrect outcome has a higher count than $T_1$. We can bound this using a union bound over all possible incorrect outcomes:
 $$ \prob{t_1^* \ne T_1} = \prob{\exists j \ne T_1, N_j \ge N_{T_1}} \le \sum_{j \ne T_1} \prob{N_j \ge N_{T_1}} \le \sum_{j \ne T_1} \exp(-\NumShots \Delta_{T_1,j}^2/2). $$
 This sum is dominated by the term with the smallest probability gap $\Delta_{T_1,j}$. We make a key distinction here: the smallest gap occurs for the adjacent outcome $T_2$, but an error where $t_1^* = T_2$ is not an outright failure, as it falls under the ambiguity detection case. We, therefore, bound the probability of $t_1^*$ being any outcome other than $T_1$ or $T_2$.
 
 The probability gap for non-adjacent outcomes is bounded below by a constant. From the properties of the squared Dirichlet kernel it is known that:
 \begin{enumerate}[label=(\arabic*)]
  \item The probability of the true peak is lower-bounded by $P(T_1|\delta) = \frac{1}{2^{2m}}\cdot\frac{\sin^2(2^m \pi \theta )}{\sin^2(\pi \theta)} \ge  \frac{1}{2^{2m}}\cdot\frac{\sin^2(2^m \pi \theta )}{(\pi \theta)^2} \ge  \frac{4}{\pi^2} \approx 0.405$. Here $\theta = \delta-\frac{j}{2^m}$.
  \item For any non-adjacent outcome $j$ (i.e., $|2^\delta-j| \ge 1.5$), its probability is bounded above by $P(j|\delta) \le \frac{1}{\pi^2(1.5)^2} \approx 0.045$.
 \end{enumerate}
 This establishes a constant lower bound for the probability gap $\Delta P_{\min} = P(T_1|\delta) - \max_{j \notin \{T_1, T_2\}} P(j|\delta) \ge 4/\pi^2 - 1/(1.5\pi)^2 \approx 0.36$.
 
 Since there are at most $2^m-2$ such non-adjacent outcomes, we can bound the total error probability as follows.
 $$ \prob{t_1^* \ne T_1 \text{ and } t_1^* \ne T_2} \le (2^m-2) \exp(-\NumShots \Delta P_{\min}^2 / 2) .$$
 To ensure this probability is less than a desired error probability $\epsilon_1$, we set
 $$ (2^{m} - 2) \exp(-\NumShots \Delta P_{\min}^2 / 2) \le \epsilon_1 .$$
 Solving for $\NumShots$ yields the stated result
 $$ \NumShots \ge \frac{2}{\Delta P_{\min}^2} \ln\left(\frac{2^m - 2}{\epsilon_1}\right)  \approx 
 15.41 \ln\left(\frac{2^m - 2}{\epsilon_1}\right) .$$
 This concludes the proof. The dependence on $m$ is logarithmic, confirming that $\NumShots = O(m + \log(1/\epsilon_1))$.
  \end{proof}
 
 \item \textbf{Probabilistic Guarantees for Ambiguity Detection}

 Let $R_{true} = P(T_2|\delta)/P(T_1|\delta)$ be the true ratio of probabilities between the second and first most probable outcomes. We use a threshold $\Threshold \in (0, 1)$ to detect ambiguity.
 
 \begin{enumerate} [label=(\roman*), itemsep=0pt]
  \item \textbf{ Unambiguous Detection:}
  If the true phase is sufficiently far from an ambiguous point such that $R_{true} \le \Threshold - \Delta_R$ for some margin $\Delta_R > 0$, the algorithm will correctly identify the situation as \enquote*{not ambiguous} (i.e., $N_{t_2^*}/N_{t_1^*} \le \Threshold$) with high probability.
  
\begin{proof}
  We assume that $t_1^*=T_1$ and $t_2^*=T_2$, which holds with high probability as established in part 1. Our goal is to bound the probability of a false positive, i.e., $\prob{N_{T_2}/N_{T_1} > \Threshold}$.
  
  Let us consider the empirical frequencies (proportions) $\hat{p}_j = N_j/\NumShots$. We are interested in bounding $\prob{\hat{p}_{T_2}/\hat{p}_{T_1} > \Threshold}$, which is equivalent to $\prob{\hat{p}_{T_2} - \Threshold \hat{p}_{T_1} > 0}$. Let $p_1 = P(T_1|\delta)$ and $p_2 = P(T_2|\delta)$.
  
  Let us define a sequence of i.i.d. random variables for each measurement $k$:
  $$ X_k = \indicator{\text{outcome } k \text{ is } T_2} - \Threshold \cdot \indicator{\text{outcome } k \text{ is } T_1} .$$
  The sum of these variables, $Z' = \sum_{k=1}^{\NumShots} X_k$, is related to our empirical frequencies by $Z' = \NumShots(\hat{p}_{T_2} - \Threshold \hat{p}_{T_1})$.
  The expected value of $Z'$ is
  $$ E[Z'] = \sum_{k=1}^{\NumShots} E[X_k] = \NumShots(p_{T_2} - \Threshold p_{T_1}) .$$
  Given the condition $p_{T_2}/p_{T_1} \le \Threshold - \Delta_R$, we have $p_{T_2} \le p_{T_1}(\Threshold - \Delta_R)$. On substituting this into the expectation we obtain
  $$ E[Z'] \le \NumShots(p_{T_1}(\Threshold - \Delta_R) - \Threshold p_{T_1}) = -\NumShots p_{T_1} \Delta_R .$$
  Since $E[Z'] < 0$, the event $Z' > 0$ is a large deviation. We can apply Hoeffding's inequality. The random variable $X_k$ is bounded. Its maximum value is $1-\Threshold \cdot 0 = 1$ and its minimum value is $0-\Threshold \cdot 1 = -\Threshold$. Thus, $X_k \in [-\Threshold, 1]$, with a range of $1-(-\Threshold) = 1+\Threshold$.
  
  Applying Hoeffding's inequality gives
  \begin{align*}
  \prob{\hat{p}_{T_2} - \Threshold \hat{p}_{T_1} > 0} &= \prob{Z' > 0} \\
  &= \prob{Z' - E[Z'] > -E[Z']} \\
  &\le \exp\left(-\frac{2(E[Z'])^2}{\NumShots (1+\Threshold)^2}\right) \\
  &\le \exp\left(-\frac{2(\NumShots p_{T_1} \Delta_R)^2}{\NumShots (1+\Threshold)^2}\right) \\
  &= \exp\left(-\frac{2\NumShots p_{T_1}^2 \Delta_R^2}{(1+\Threshold)^2}\right).
  \end{align*}
  To ensure this probability is less than $\epsilon_2$, we set
  $$ \exp\left(-\frac{2\NumShots p_{T_1}^2 \Delta_R^2}{(1+\Threshold)^2}\right) \le \epsilon_2 .$$
  Form which it follows that
  $$ \NumShots \ge \frac{(1+\Threshold)^2}{2p_{T_1}^2 \Delta_R^2} \ln\left(\frac{1}{\epsilon_2}\right) .$$
  This provides a quantitative bound on the number of shots required for unambiguous detection.
     
\end{proof}

  \item \textbf{ Ambiguous Detection:}
  If the true phase is very close to an ambiguous point such that $R_{true} \ge \Threshold + \Delta_R$ for some margin $\Delta_R > 0$, the algorithm will correctly identify the situation as `ambiguous' (i.e., $N_{t_2^*}/N_{t_1^*} > \Threshold$) with high probability.

\begin{proof}
  This proof is symmetric to Case (i). We are interested in bounding the probability of a false negative: $\prob{N_{T_2}/N_{T_1} \le \Threshold}$, which is equivalent to $\prob{\hat{p}_{T_2} - \Threshold \hat{p}_{T_1} \le 0}$ or $Z' \le 0$.
  
  The condition is $p_{T_2}/p_{T_1} \ge \Threshold + \Delta_R$. The expected value of $Z'$ is now positive and we have
  $$ E[Z'] = \NumShots(p_{T_2} - \Threshold p_{T_1}) \ge \NumShots(p_{T_1}(\Threshold + \Delta_R) - \Threshold p_{T_1}) = \NumShots p_{T_1} \Delta_R .$$
  The probability of a false negative is $\prob{Z' \le 0} = \prob{Z' - E[Z'] \le -E[Z']}$.
  Applying Hoeffding's inequality with the same bounds for $X_k$ gives
  $$ \prob{Z' \le 0} \le \exp\left(-\frac{2(E[Z'])^2}{\NumShots (1+\Threshold)^2}\right) \le \exp\left(-\frac{2(\NumShots p_{T_1} \Delta_R)^2}{\NumShots (1+\Threshold)^2}\right) = \exp\left(-\frac{2\NumShots p_{T_1}^2 \Delta_R^2}{(1+\Threshold)^2}\right) .$$
  To ensure this probability is less than $\epsilon_2$, we obtain the same bound on $\NumShots$. The identification of $t_1^*$ and $t_2^*$ with $T_1$ and $T_2$ is guaranteed by part 1 of this lemma.
\end{proof}
 \end{enumerate}
\end{enumerate}

\begin{remark}
\label{rem:robust_ambiguity_detection}
\textbf{The Robustness of the Ambiguity Detection Mechanism}

The choice of the ambiguity threshold $\Threshold$ is a critical design parameter for the \AWQPEs \, algorithm. Our proof establishes a sufficient condition for the number of shots $\NumShots$ to reliably distinguish between two well-separated cases: when the true probability ratio $R_{true} = P(T_2|\delta)/P(T_1|\delta)$ is far from the threshold (i.e., $R_{true} \leq \Threshold - \Delta_R$ or $R_{true} \geq \Threshold + \Delta_R$), where $\Delta_R > 0$ is a theoretical margin.

However, the algorithm's practical performance is more robust than this proof requires, thanks to its integrated design:

\begin{enumerate} 
\item \textbf{High $\Threshold$ as a Safe Design Choice:} By choosing a high $\Threshold$ (e.g., $0.9$), we conservatively define ambiguity. This ensures that the ambiguity flag is only set for genuinely ambiguous phases where $R_{true}$ is very close to 1. This minimizes the risk of false positives, where a non-ambiguous phase is incorrectly flagged as ambiguous. It also allows the LSB-to-MSB post-processing routine to handle the majority of rounding corrections, leading to a cleaner separation of concerns.

    \item \textbf{Resilience to Low $\Threshold$ (Preemptive Correction):} As demonstrated by analysis of the algorithm's behavior, even if a non-ambiguous phase is incorrectly flagged as ambiguous due to a low $\Threshold$ (e.g., $0.5$), the algorithm is often resilient to this misclassification. In such cases, Algorithm 1 selects the raw chunk estimate using $\min(t_1^*, t_2^*)$, which effectively performs a ``round down" or ``borrow" correction preemptively. Since Algorithm 2 then disables further correction for that chunk (due to the ambiguity flag), this early correction still leads to the correct final phase estimate. This shows that the algorithm possesses a form of redundant or parallel correction logic, making it robust to a wide range of $\Threshold$ values.
\end{enumerate}

The choice of $\Threshold$ sets the algorithm's operational definition of ambiguity, while $\Delta_R$ is a parameter used in the proof to quantify the number of shots needed for a desired level of statistical confidence. The existence of a positive $\Delta_R$ ensures that the proof is valid, but the practical effectiveness of the algorithm stems from its well-designed handling of all possible scenarios, including those in the ``grey area" where the true probability ratio falls between $\Threshold - \Delta_R$ and $\Threshold + \Delta_R$. 
\end{remark}

\begin{remark}
The number of measurement shots required in \AWQPEs \, scales as
\[
\NumShots = O\left(\frac{1}{\Delta^2} \log\left(\frac{1}{\epsilon}\right)\right),
\]
where $\Delta$ refers either to the probability gap between outcomes (for correctness) or the margin $\Delta_R$ from the ambiguity threshold (for ambiguity detection). In particular, the number of shots needed to reliably detect whether $N_{T_2}/N_{T_1} \gtrless \Threshold$ scales as $O\left(\frac{1}{\Delta_R^2} \log\left(\frac{1}{\epsilon_2}\right)\right)$.

We use Hoeffding-style concentration bounds to derive these theoretical guarantees. These bounds are simple, variance-independent, and apply universally to bounded i.i.d. variables. Stronger alternatives, such as \textbf{Bernstein's inequality} (which incorporates variance information) or \textbf{Sanov-type large deviation bounds} (for multinomial distributions), can yield tighter constants in the exponential decay of error probabilities.

In particular, Bernstein bounds are useful when the dominant QPE outcome has high probability and low variance, as is typical in Dirichlet-like distributions. However, all such bounds preserve the same asymptotic scaling, so our theoretical complexity analysis remains valid. These refined bounds may offer practical benefits in implementations seeking to reduce measurement overhead.
\end{remark}

\begin{remark}
Sanov’s theorem provides asymptotically optimal estimates for the probability that the empirical distribution $\hat{P}_n$ deviates from the true QPE outcome distribution $P(j|\delta)$. Specifically, the probability that the most frequent empirical outcome differs from the true maximum $T_1$ satisfies:
\[
\mathbb{P}(\hat{P}_n \in \mathcal{A}) \le \exp\left(-\NumShots \cdot D_{\mathrm{KL}}^*\right),
\]
where $\mathcal{A}$ is the set of empirical distributions for which $\arg\max_j \hat{P}_n(j) \ne T_1$, and $D_{\mathrm{KL}}^*$ is the minimal KL divergence over this set.

While powerful in theory, Sanov bounds are nonconstructive in practice and do not yield explicit expressions for $\NumShots$ without computing $D_{\mathrm{KL}}^*$. For this reason, we rely on Hoeffding and Bernstein bounds for tractable and implementation-friendly shot complexity estimates.
\end{remark}

With the preceding lemmas establishing the necessary mathematical groundwork for handling binary approximations and their compositions, we are now equipped to present the comprehensive proof of correctness for the Adaptive Windowed Quantum Phase Estimation algorithm. This proof will demonstrate how the iterative quantum measurements, combined with the described classical post-processing and ambiguity resolution, reliably converge to the true phase $\phi$.

\subsection{Proof of \AWQPEs \, Ambiguity Resolution with Probabilistic Guarantees}

\begin{thm}[Accuracy and Ambiguity Resolution of AWQPE]
\label{thm:awqpe_accuracy_resolution_final}
Let $U$ be a unitary operator with an eigenstate $\ket{\psi}$ corresponding to an eigenvalue $e^{2\pi i \phi}$, where $\phi \in [0, 1)$ is the unknown phase. \AWQPE \, (\AWQPEs) algorithm, utilizing $\nTotal$ total control qubits distributed across $B$  blocks (e.g., $\mList$) produces an estimate $\tilde{\phi}$ of $\phi$.
For a sufficiently large number of control qubits $n$, the \AWQPE \, algorithm successfully determines the best $n$-bit approximation of $\phi$, i.e., $\lfloor 2^n \phi + 0.5 \rfloor / 2^n$ with high probability. 
\end{thm}

\begin{proof}
It is known that the standard QPE with $m$ control qubits  produces measurement outcomes that results in the best $m$-bit approximation of $\phi$. 
Before showing the correctness of the LSB-to-MSB correction approach in \AWQPEs \, algorithm, we must establish that the raw binary chunks ($\phiRaw$)  are reliable. The AWQPE algorithm's core is the measurement of a window of $m_j$ control qubits in the $j$-th iteration. For each iteration, the true outcome probabilities follow a squared Dirichlet kernel distribution, with a true most probable outcome, $T_1$.

Lemma \ref{lem:probabilistic_outcome_quantitative} guarantees that, with a sufficient number of measurement shots ($\NumShots$), the empirically observed most frequent outcome $t_1^*$ will match the true most probable outcome $T_1$ with a high probability of at least $1-\epsilon_1$. This establishes that the raw chunk estimates, $\phi^{(j)}$, are reliable approximations of the true phase bits within their window, with a high degree of confidence.

Furthermore, Lemma \ref{lem:probabilistic_outcome_quantitative} also provides probabilistic guarantees for the ambiguity detection mechanism. If the true phase is close to an ambiguous point, the algorithm will correctly flag it as ambiguous with high probability. Conversely, if the true phase is non-ambiguous, it will be correctly identified as such. This confirms that the ambiguity flags ($\AmbFlags$) used in the post-processing step are also probabilistically reliable.

Based on these guarantees, we can now proceed with the induction argument on the post-processing routine, assuming the raw chunks $\phi^{(j)}$ and their corresponding ambiguity flags are correct with high probability.

The correctness of the LSB-to-MSB correction routine follows by an straightforward induction argument, working backward from the least significant chunk to the most significant. The base case (the final chunk) requires no correction as there are no further LSBs. For the inductive step, the correction rule, which applies a ``borrow" operation based on the MSB of the next less significant chunk, is directly justified by the principles established in Lemma \ref{lem:combined_approx}. This lemma precisely dictates how to combine best approximations of a value and its fractional part. The algorithm's specific handling of ambiguous points and ``special chunks" (e.g., fractional part of 0.5) also aligns with the conditions and considerations detailed within Lemma \ref{lem:combined_approx}, preventing incorrect carry/borrow propagation. Thus, by iteratively applying this logic, the algorithm correctly reconstructs the true phase estimate from the measured chunks. Full details are omitted for brevity.
\end{proof}

\section{Computational Complexity Analysis}
\label{sec:complexity}
This section provides a detailed comparative analysis of the computational complexity of the Adaptive Windowed Quantum Phase Estimation (AWQPE) algorithm and the Standard Quantum Phase Estimation (Standard QPE) algorithm. The analysis focuses on key resource requirements for implementation on noisy quantum hardware, including qubit count, circuit depth, gate count, and the number of measurement shots.

For this analysis, let $n$ be the number of bits in the phase estimate and $\NumTarg$ be the number of target qubits. We denote the gate complexity of a single unitary $U$ as $C_g(U)$ and its circuit depth as $C_d(U)$.

\subsection{Standard Quantum Phase Estimation (Standard QPE)}
The standard QPE algorithm estimates an $n$-bit phase by executing a single, deep quantum circuit. The algorithm requires $n$ auxiliary control qubits to store the phase estimate, in addition to the $\NumTarg$ qubits for the unitary operator $U$. The total qubit count is therefore $n + \NumTarg$. For high-precision estimates, this linear scaling with the desired precision is a significant and often prohibitive resource requirement for current quantum hardware.

The circuit depth is primarily determined by the sequence of controlled unitary operations. This sequence involves controlled-$U^{2^j}$ operations for $j=0, \dots, n-1$. A controlled-$U^{2^j}$ operation is typically synthesized by applying the controlled-$U$ operation $2^j$ times. Consequently, the depth of a single controlled-$U^{2^j}$ scales as $O(2^j \cdot C_d(U))$. The total depth of the controlled unitary sequence is dominated by the largest term, leading to a total circuit depth that scales as $O(2^n \cdot C_d(U))$. This exceptionally deep circuit architecture is highly susceptible to decoherence and error accumulation.

The total gate count for standard QPE has two primary components. The Inverse Quantum Fourier Transform (IQFT) on $n$ qubits requires a gate count of $O(n^2)$. The controlled unitary part of the circuit, however, consists of $n$ controlled-$U^{2^j}$ gates. The total number of applications of the unitary $U$ that must be synthesized into the circuit is $\sum_{j=0}^{n-1} 2^j = 2^n - 1$. This makes the controlled unitary part the dominant factor in the overall gate count, which scales as $O(2^n \cdot C_g(U))$. The total gate complexity of standard QPE is therefore dominated by this exponential term.

\subsection{Adaptive Windowed Quantum Phase Estimation (AWQPE)}
The AWQPE algorithm addresses the limitations of standard QPE by iteratively estimating the phase in smaller chunks, or ``windows," of size $m_i$. Let $B$ be the number of windows, such that the total precision is $n = \sum_{i=1}^B m_i$.

The number of control qubits needed at any point in time is only the size of the current window, $m_i$. By choosing small, fixed window sizes, the algorithm's control qubit requirement is reduced to $\max(m_i)$, which can be significantly smaller than $n$.

The maximum circuit depth required at any single point in the algorithm is determined by the largest window size, $\max(m_i)$. The depth for each iteration is dominated by the controlled-$U^{2^j}$ operations. For the $i$-th chunk, which estimates bits from $p_{\text{start}} = \sum_{j=1}^{i-1} m_j$ to $p_{\text{end}} = \sum_{j=1}^{i} m_j - 1$, the controlled unitary operations are of the form $U^{2^{p_{\text{start}}+p}}$ for $p=0, \dots, m_i-1$. The depth of the circuit for this chunk is therefore dominated by the final controlled unitary operation, which has a depth of $O(2^{p_{\text{start}}+m_i-1} \cdot C_d(U))$. This approach distributes the total computational effort across multiple shallow circuits (or independent blocks), making the algorithm highly resilient to decoherence compared to standard QPE.

The total gate count for AWQPE is the sum of the gate counts for each of the $B$ chunks. Each chunk requires an IQFT on $m_i$ qubits, with a gate count of $O(m_i^2)$. For the $i$-th chunk, the total number of applications of $U$ that must be synthesized into the circuit is $\sum_{p=0}^{m_i-1} 2^{p + \sum_{j=1}^{i-1} m_j} = 2^{\sum_{j=1}^{i-1} m_j} \cdot (2^{m_i}-1) $.
The total number of applications of $U$ across all chunks is the sum of these terms and it is equal to $ 2^n -1$. 
A crucial advantage of this iterative approach is that the individual blocks can be executed independently, significantly reducing the depth and resource requirements of any single execution. The overhead of the classical post-processing step is negligible in comparison to the quantum cost.

\subsection{Summary of Comparison}
Table \ref{tab:complexity_comparison} provides a comprehensive summary of the key trade-offs between the two algorithms. The AWQPE algorithm offers a more pragmatic approach for practical phase estimation on near-term quantum hardware by reducing the control qubit count, circuit depth, and the number of gates required for any single quantum circuit execution.

\begin{table}[h!]
\centering
\caption{Comparison of computational resources for standard QPE and a single circuit block of AWQPE.}
\label{tab:complexity_comparison}
\begin{tabular}{lcc}
\toprule
\textbf{Metric} & \textbf{Standard QPE} & \textbf{AWQPE ($i$-th Block)} \\
\midrule
Control Qubits & $n$ & $m_i$ \\
Circuit Depth & $O(2^n \cdot C_d(U))$ & $O(2^{\sum_{j=1}^{i} m_j - 1} \cdot C_d(U))$ \\
Total Applications of $U$ & $O(2^n)$ & $2^{\sum_{j=1}^{i-1} m_j} \cdot (2^{m_i}-1)$ \\
\bottomrule
\end{tabular}
\end{table}

In summary, the Adaptive Windowed Quantum Phase Estimation (AWQPE) algorithm presents a powerful approach for quantum phase estimation on near-term hardware. By decomposing the estimation of a high-precision phase into a series of smaller, iterative sub-problems, AWQPE mitigates the exponential resource demands of the standard QPE algorithm. This method effectively redistributes the computational burden from a single, resource-intensive circuit to multiple shallower circuits, significantly reducing the required control qubit count and maximum circuit depth for any single coherent operation. Consequently, AWQPE emerges as a viable and efficient method for platforms where large-scale, high-coherence circuits are not yet feasible, while still achieving the desired level of precision. The algorithm's flexibility in choosing window sizes, $m_j$, allows for further optimization tailored to the specific capabilities and limitations of the target quantum hardware.

\begin{remark}
Various QPE variants in the literature aim to reduce the unitary construction complexity from \( O(2^n) \). Our scheme is compatible with such approaches and stands to benefit from the resulting reduction in unitary gate complexity.
\end{remark}

\section{Advantages of Adaptive Windowed Quantum Phase Estimation}
\label{sec:advantages}

The Adaptive Windowed Quantum Phase Estimation (AWQPE) algorithm offers several significant advantages that position it as a highly promising approach for quantum phase estimation, particularly on contemporary and near-term quantum hardware. These benefits stem from its modular design and robust classical post-processing.

A primary advantage of AWQPE lies in its \textit{modular decomposition of phase estimation}. The overall phase estimation task is broken down into multiple independent segments, where each segment estimates a specific window of $m$ bits from the binary expansion of the phase $\phi$. Specifically, for a given block $b$, the algorithm aims to estimate $\phi_b := 2^{bm} \phi \bmod 1$. This localized estimation strategy enables a scalable and incremental recovery of the full phase, distributing the computational burden across smaller, more manageable units.

This modularity directly translates into \textit{shorter circuits per block}. Each AWQPE block operates on only $m$ qubits, rather than requiring a full $n$-qubit register for the entire precision. This reduction in circuit width and depth significantly lowers the demands on quantum memory and coherence time, making the approach exceptionally well-suited for Noisy Intermediate-Scale Quantum (NISQ) devices and other near-term quantum processors characterized by limited resources and high error rates.

Furthermore, AWQPE incorporates a robust mechanism for \textit{tolerance to ambiguity via Most Significant Bit (MSB) correction}. Due to inherent quantum measurement uncertainties, rounding errors, and modular aliasing, the outputs of individual blocks may suffer from ambiguity in their MSBs. To counteract this, a robust classical post-processing routine propagates corrections from right-to-left by comparing adjacent blocks. If a more significant block (e.g., block $b+1$) yields an unambiguous result, its MSB can be used to disambiguate the preceding block (block $b$). Crucially, ambiguity flags are employed to ensure that unreliable updates are gracefully avoided in cases where the measurement outcomes are highly uncertain, thereby enhancing the overall reliability of the phase estimate.

The probabilistic guarantees of AWQPE are rigorously \textit{informed by the Dirichlet kernel}. The probability distribution of each block's output is accurately characterized using the squared Dirichlet kernel, centered at $\phi_b$. This robust mathematical framework enables a rigorous modeling of outcome likelihoods, allowing for the derivation of analytical bounds on the number of measurement shots required to attain a given confidence level in the estimated phase.

The inherent design of AWQPE also facilitates a \textit{parallelizable and segmentable architecture}. Since individual modular QPE blocks operate independently, they can be executed concurrently across multiple quantum processors. This segmentable structure is particularly ideal for distributed quantum computing platforms, such as those based on photonic, ion-trap, or cryogenic qubit arrays, where independent execution and localized control offer significant operational benefits.

Moreover, the modularity confers \textit{resilience to local gate errors}. Errors occurring within a given block do not propagate globally throughout the entire computation. Each modular block is self-contained, meaning that faulty outputs from a specific block can be selectively discarded or remeasured without invalidating the results obtained from other blocks. This localized error containment enhances the algorithm's robustness to hardware noise and enables the application of statistical filtering or targeted repetition strategies to improve overall fidelity.

Finally, the modular structure of AWQPE is \textit{hardware-friendly for low-latency feedback}. It facilitates rapid measurement and feedback loops between quantum and classical components. Block-wise results can be validated, filtered, or fed into adaptive phase tuning procedures in a timely manner. This aligns naturally with hybrid quantum-classical workflows that increasingly rely on fast iterative updates and early stopping conditions to optimize resource utilization and achieve desired precision within the constraints of current quantum hardware.

The aforementioned theoretical advantages of the Adaptive Windowed Quantum Phase Estimation algorithm, including its modularity, resilience to errors, and efficient resource utilization, are critical for its applicability on near-term quantum hardware. To empirically validate these claims and demonstrate the practical performance of AWQPE, we now present the results from comprehensive numerical simulations. These simulations illustrate the algorithm's accuracy, efficiency, and robustness under various conditions, providing concrete evidence of its capabilities.

\section{Numerical Simulation Results}
\label{sec:results}

In this section we present numerical examples. The examples are based on a Python implementation using the IBM's \texttt{Qiskit} library \cite{Qiskit} and illustrate the algorithm's performance for various phase values and chunk sizes.

\subsection{Detailed Walkthrough: Test Case with Non-Uniform Chunk Sizes}

\begin{figure}[h!]
    \centering
    \begin{subfigure}[b]{0.9\textwidth} 
        \centering
        \includegraphics[width=0.6\textwidth]{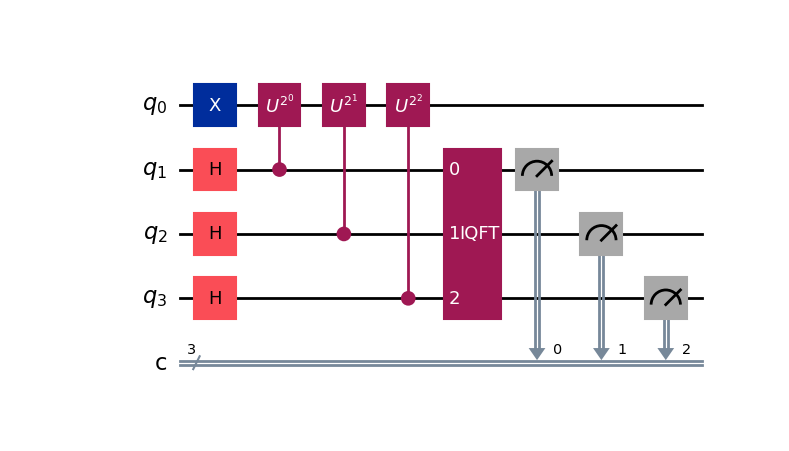} 
        \caption{Circuit for the first iteration block, estimating 3 MSBs.}
        \label{fig:circuit_iter1}
    \end{subfigure}
    \vspace{0.1cm} 

    \begin{subfigure}[b]{0.9\textwidth}
        \centering
        \includegraphics[width=0.6\textwidth]{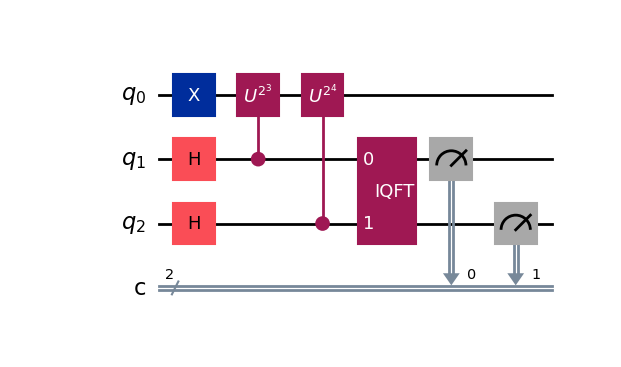} 
        \caption{Circuit for the second iteration block, estimating the next 2 bits.}
        \label{fig:circuit_iter2}
    \end{subfigure}
    \vspace{0.1cm} 

    \begin{subfigure}[b]{0.9\textwidth}
        \centering
        \includegraphics[width=0.6\textwidth]{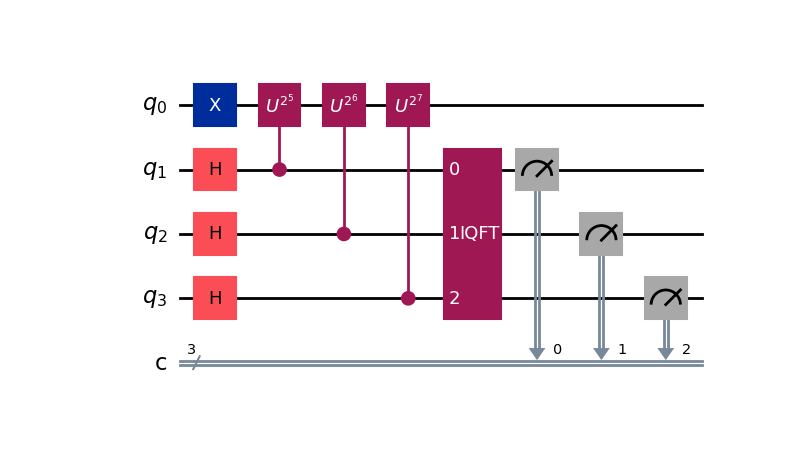} 
        \caption{Circuit for the third iteration block, estimating the final 3 bits.}
        \label{fig:circuit_iter3}
    \end{subfigure}

    \caption{Quantum circuit diagrams for the non-uniform chunk size test case ($\phi = 0.8203125$ with $\mList = [3, 2, 3]$). Each circuit corresponds to an iterative block of the \AWQPE \, algorithm.}
    \label{fig:all_circuits}
\end{figure}

\begin{figure}[h!]
    \centering
        \centering
        \includegraphics[width=0.9\textwidth]{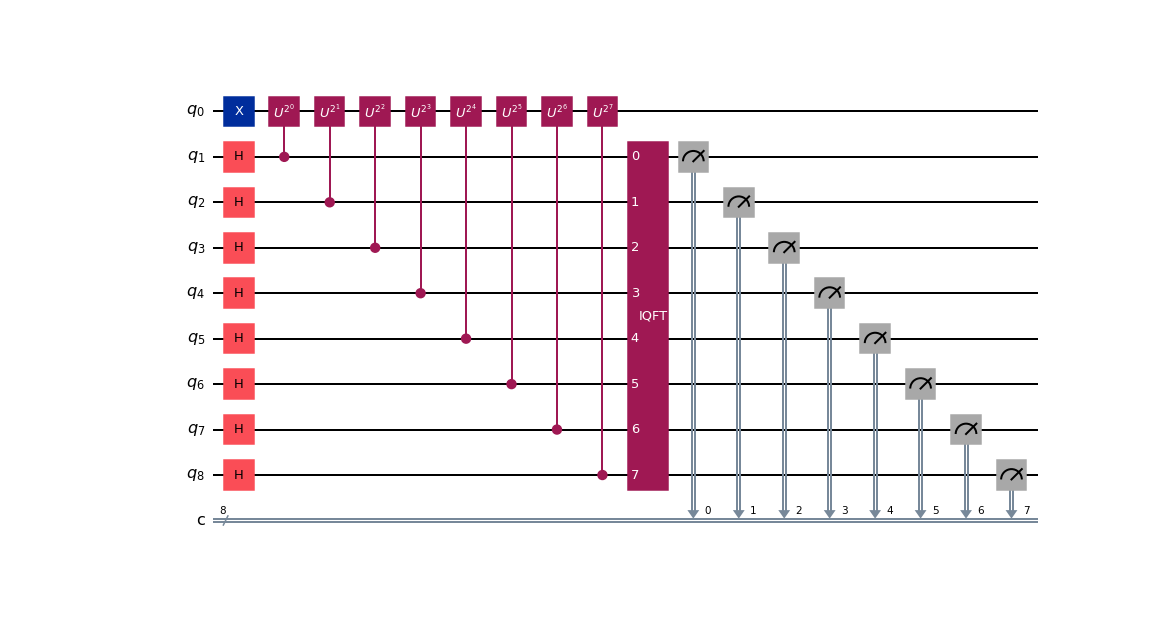} 
        \caption{Quantum circuit for the standard QPE for the  test case ($\phi = 0.8203125$ with $n = 8$ control qubits).}
        \label{fig:circuit_full}
    \end{figure}

This example demonstrates the \AWQPEs \, algorithm's ability to handle non-uniform chunk sizes and apply post-processing correction to the measured outcomes. The true phase is $\phi = 0.8203125$ with window sizes defined by $\mList = [3, 2, 3]$, resulting in a total of $\nTotal = 8$ bits to estimate. The true binary representation of $\phi$ is $0.11010010_2$.

The \AWQPEs \, algorithm proceeds in a series of iterative steps. In the first iteration, a QPE circuit is executed to estimate the first three most significant bits ($m=3$). The measurement outcomes from the simulation do not show a close tie between the top two most likely outcomes \enquote*{111} and \enquote*{110}. The ambiguity resolution mechanism of the algorithm  does not get triggered and  \enquote*{111} is selected as the most likely unambiguous outcome. The estimated binary string is therefore \enquote*{0.111}. The raw counts from the simulation were:
\begin{verbatim}
[('111', 5180), ('110', 3284), ('000', 534), ('101', 459), ('001', 237), ...].
\end{verbatim}

Next, the second QPE circuit is run to estimate the next two bits ($m=2$). The simulation yields an unambiguous measurement of `10', updating the estimated binary string to \enquote*{0.11110}. The raw counts were:
\begin{verbatim}
[('10', 8374), ('11', 1078), ('01', 467), ('00', 321)].
\end{verbatim}

For the final iteration, a QPE circuit estimates the last three bits ($m=3$), with an unambiguous measurement of \enquote*{010}. The concatenated raw binary estimate is therefore \enquote*{11110010}. The raw counts were:
\begin{verbatim}
[('010', 10240)].
\end{verbatim}

Following the measurements, a post-processing and correction procedure is applied from the least significant bit (LSB) to the most significant bit (MSB). For the second-to-last chunk (\enquote*{10}), the most significant bit (MSB) of the subsequent chunk (\enquote*{010}) is 0. As per the algorithm, no correction is applied, and the chunk value remains \enquote*{10} ($(2 - 0) \pmod{4} = 2$). For the first chunk (\enquote*{111}), the MSB of the next chunk (\enquote*{10}) is 1. Further, since an ambiguity flag was not set (and no special chunk was obtained) for the previous chunk, the subtraction bit is considered $1$, and so a bit-correction is applied. The chunk value changes to \enquote*{110} ($(7 - 1) \pmod{8} = 6$).

The final corrected binary estimate is \enquote*{0.11010010}, which corresponds to a decimal value of $0.8203125$. The algorithm successfully produced an accurate estimate of the true phase.

The distinct circuit configurations for each iterative block of the AWQPE algorithm for the above example are illustrated in Figure~\ref{fig:all_circuits}, with the quantum circuit for the first iteration (estimating the most significant bits) shown in Figure~\ref{fig:circuit_iter1}, the subsequent quantum circuit for the middle chunk in Figure~\ref{fig:circuit_iter2}, and the quantum circuit for the last chunk of bits in Figure~\ref{fig:circuit_iter3}. The quantum circuit for the standard QPE with $n=8$ control qubits is shown in Figure~\ref{fig:circuit_full}. It is clear that \AWQPEs \, algorithm produces considerably shallower quantum circuits in comparison to the standard QPE.

The performance of the \AWQPEs \, algorithm was rigorously evaluated across a wide range of test cases. Using the Qiskit simulation environment, thousands of diverse scenarios were examined, encompassing varied input phase values, non-uniform chunk sizes, and randomly selected eigenvalues/phases. In each instance, the estimated phase values obtained through the \AWQPEs \, algorithm demonstrated consistent accuracy. A selection of representative test cases is summarized in Table \ref{tab:summary}, clearly illustrating the algorithm's effectiveness. These results confirm that the \AWQPEs \, algorithm, including its post-processing correction mechanism, reliably provides precise phase estimates. In addition to Qiskit simulations, the algorithm's performance was further validated using a direct Dirichlet kernel based simulation, which was applied to over a million randomly generated test cases, yielding accurate results consistent with theoretical expectations.

\begin{table*}[htbp]
\centering
\caption{Summary of \AWQPEs \, Numerical Examples. }
\label{tab:summary}
\resizebox{\textwidth}{!}{
\begin{tabular}{ccccc}
\toprule
\textbf{Test Case} & \textbf{Input Parameters} ($\phi$, $\mList$) & \textbf{Raw Binary Estimate} & \textbf{Final Binary Estimate} & \textbf{Final Decimal Estimate} \\
\midrule
1 & $\phi=0.3$, $m=[2, 2]$ & \texttt{0101} & \texttt{0101} & 0.3125 \\
2 & $\phi=\pi/6$, $m=[3, 2, 2, 3]$ & \texttt{1000111000} & \texttt{1000011000} & 0.5234375 \\
3 & $\phi=0.671875$, $m=[4, 4]$ & \texttt{10111100} & \texttt{10101100} & 0.671875 \\
4 & $\phi=1/\sqrt{2}$, $m=[3, 3, 3, 3, 3, 3, 3, 3, 3, 3]$ & \texttt{110101010000010100110011010101} & \texttt{101101010000010011110011001101} & 0.7071067811921239 \\
5 & $\phi=\sin(\pi/12)$, $m=[5, 6, 7, 4]$ & \texttt{0100001001000010001110} & \texttt{0100001001000001111110} & 0.2588191032409668 \\
\bottomrule
\end{tabular}
}
\end{table*}

\section{Conclusion}
\label{sec:conclusion}

In this work, we have introduced the Adaptive Windowed Quantum Phase Estimation (AWQPE) algorithm, a modular and scalable approach that offers a compelling advantage over existing methods. AWQPE directly addresses the resource-intensive nature of the standard QPE algorithm, which is largely impractical for current Noisy Intermediate-Scale Quantum (NISQ) devices due to its deep circuits and high qubit count. By estimating multiple phase bits (of independent blocks) simultaneously in a ``window" of control qubits   AWQPE drastically reduces resources required to achieve a target precision. These independent blocks are amenable for parallelization, which allows for a shallower overall circuit and a significantly lower total execution time, making it a far more viable solution for near-term quantum hardware than the standard QPE approach.

A central feature of AWQPE is its robust classical post-processing logic, which not only resolves ambiguities in the phase estimation but also corrects for measurement errors. This mechanism enhances the reliability of the final estimate, a crucial factor in the presence of noise. 
Our numerical simulations confirm the algorithm's effectiveness and its superior performance in comparison to standard QPE, demonstrating the potential for realizable and efficient implementations on near-term quantum platforms.

\bibliographystyle{unsrt}

\end{document}